\documentclass[aip,jmp,reprint,onecolumn]{revtex4-1}
\usepackage{amsmath,amsfonts,amssymb,amsthm,stmaryrd,bm}
\usepackage{mathtools}
\usepackage{wasysym}
\usepackage{graphicx}
\usepackage{hyperref}
\usepackage{breakurl}
\usepackage{multirow}
\usepackage{color}

\newcommand{\Dslash}{{\slash{\kern -0.5em}\partial}}
\newcommand{\Aslash}{{\slash{\kern -0.5em}A}}

\def\sqr#1#2{{\vcenter{\hrule height.#2pt
     \hbox{\vrule width.#2pt height#1pt \kern#1pt
        \vrule width.#2pt}
     \hrule height.#2pt}}}

\def\thinspace{\kern .16667em}

\def\xp{x_{{\kern -.2em}_\perp}}
\def\subp{_{{\kern -.2em}_\perp}}

\def\defeq{:{\kern -0.5em}=}

\DeclareMathOperator{\re}{\mathrm{Re}}

\DeclareMathOperator*{\essinf}{ess\,inf}

\newcommand{\Hilb}{{\mathcal H}}

\newcommand{\Real}{{\mathbb R}}
\newcommand{\Nat}{{\mathbb N}}
\newcommand{\Ints}{{\mathbb Z}}

\newcommand{\Num}{\mathcal N}

\newcommand\dens{\mathsf{dens}\,}

\newcommand\EE{{\mathcal E}}

\newcommand{\VV}{{\cal V}}
\newcommand{\tloc}{{\mathrm{loc}}}

\newcommand{\Leb}[1]{\mathrm{Leb}\, #1}
\newcommand{\vext}{v_{\text{ext}}}
\newcommand{\Vext}{V_{\text{ext}}}
\newcommand{\Potl}{{\VV}_{\text{ext}}}
\newcommand{\vint}{v_{\text{int}}}
\newcommand{\Vint}{V_{\text{int}}}
\newcommand{\Vtot}{V_{\text{tot}}}
\newcommand{\setof}[2]{\left\{ {#1} \;|\; {#2} \right\}}
\DeclareMathOperator{\Md}{SM}
\DeclareMathOperator{\Tm}{LI}

\DeclareMathOperator{\intr}{int}
\DeclareMathOperator{\bndy}{bndy}
\DeclareMathOperator{\cl}{cl}
\DeclareMathOperator{\comps}{conn}
\newcommand{\Univ}{{\mathbb U}}
\newcommand{\cequiv}{\,{\stackrel{c}{\sim}}\,}

\newcommand{\inpr}[2]{\left\langle {#1} \middle| {#2} \right\rangle}
\newcommand{\outpr}[2]{\left| {#1} \middle\rangle\middle\langle {#2} \right|}
\newcommand{\DiracProj}[1]{\outpr{{#1}}{{#1}}}

\newcommand{\Dbracket}[3]{\left\langle{#1}\middle|{#2}\middle|{#3}\right\rangle}

\theoremstyle{plain}
\newtheorem{thm}{Theorem}[section]
\newtheorem{lem}[thm]{Lemma}
\newtheorem{cor}[thm]{Corollary}
\newtheorem{prop}[thm]{Proposition}
\newtheorem*{lem*}{Lemma}
\newtheorem*{cor*}{Corollary}
\newtheorem*{thm*}{Theorem}

\theoremstyle{definition}
\newtheorem{defn}{Definition}[section]

\theoremstyle{remark}

\theoremstyle{definition}

\begin{document}

\title{In search of the Hohenberg-Kohn theorem}
\author{Paul~E.~Lammert}
\email{lammert@psu.edu}
\affiliation{Department of Physics, 104B Davey Lab \\ 
Pennsylvania State University \\ University Park, PA 16802-6300}
\date{March 28, 2018}
\begin{abstract}
The Hohenberg-Kohn theorem, a cornerstone of electronic 
density functional theory, concerns uniqueness of external potentials
yielding given ground densities of an ${\mathcal N}$-body system.
The problem is rigorously explored in a universe of three-dimensional 
Kato-class potentials, with emphasis on 
trade-offs between conditions on the density and conditions on 
the potential sufficient to ensure uniqueness.
Sufficient conditions range from none on potentials
coupled with everywhere strict positivity of the density, to
none on the density coupled with something a little weaker than
local $3{\mathcal N}/2$-power integrability of the potential on a
connected full-measure set.
A second theme is localizability, that is, the possibility of
uniqueness over subsets of ${\mathbb R}^3$ under less stringent conditions.
\end{abstract}
\maketitle

\section{Introduction}
\label{sec:intro}

The Hohenberg-Kohn (HK) theorem~\cite{HK} is a basic principle of density 
functional theory (DFT), addressing,
for a system of $\Num$ identical particles with interaction $\vint$, 
the question of how many external one-body potentials $\vext$ have a ground 
state with given one-body density $\rho(x)$. (Such a potential is said to
{\it represent} the density $\rho$.)
The answer it gives is that, if $\rho$ has any representing potential,
it is essentially unique (almost everywhere and up to a global constant).
This is considered a cornerstone of DFT since, as Martin~\cite{Martin} puts it,
``Therefore all properties of the system are completely determined given
 only the ground-density.''
The only known proof strategy has two steps. Step 1: show that if two potentials
share a ground density, then they share a ground state. This is unproblematic.
Step 2: show that there can be no common eigenstate unless the potentials differ
by only a constant. There are a couple of heuristic approaches to this 
Step\cite{HK,Parr+Yang,Martin,Dreizler+Gross,Eschrig,Engel+Dreizler}.
However, the DFT literature does not seem to contain a thoroughly justified, 
non-vague statement of the theorem's scope. 

This paper fills that gap, not with a single ``HK theorem'', but several.
A major theme here is trade-offs between conditions on potentials and on the density
sufficient to guarantee uniqueness, as highlighted by the following Omnibus HK theorem.
Within a large universe of potentials (containing the Kato class $K_{3}$),
and assuming $\vint$ is the usual Coulomb interaction for simplicity,
$\rho$ has an essentially unique representing potential (if any), in case
(Core) $\rho$ vanishes nowhere; 
(Weak) $\rho$ vanishes on a set of measure zero and
some representing potential $\vext$ satisfies a {\it weak} condition;
or (Strong) some representing potential $\vext$ satisfies a {\it strong} condition.
The {\it weak} and {\it strong} conditions are of the form,
`$\vext$ is ``nice'' on a ``large'' set'. 
The {\it weak} condition involves the slightly esoteric notion of 
Sobolev multipliers\cite{Mazya-85-book,Mazya-09-book}. A simplified version is:
$\vext$ is weakly-$L^3$ (this allows $1/r$ Coulomb singularities) near every point
of a connected open dense set. The {\it strong} condition adds a layer to this.
Again, a simplified version is: $|\vext|^{3\Num/2}$ is locally integrable near
every point of a connected open set of full measure (zero measure complement).

It is unclear how strong a Hohenberg-Kohn theorem DFT needs.
Better put, it is unclear what r\^ole such a theorem should play.
The class of ``physical'' potentials is very limited from a mathematical perspective, 
containing at worst a discrete set of Coulomb singularities.
With $\vint$ also a Coulomb interaction, these cases are all within the
purview of even the Simple HK theorem \ref{thm:Simple HK}, which shows that
$\rho$ is almost everywhere nonzero, and full uniqueness holds.
Concerning the density with a molecular potential, much more refined 
results are available\cite{Fournais+-04,Fournais+09}
such as analyticity away from the nuclei.
However, DFT is committed to considering as well ``nonphysical'' potentials.
For instance, Lieb's theory\cite{Lieb83}, widely adopted as 
a framework, uses the very large space $L^{3/2}(\Real^3) + L^\infty(\Real^3)$.
The uncertainty of the situation justifies considering a range of 
results as done here.

Conditions in every case of the Omnibus HK theorem require nice behavior on
a large set. Cases Weak and Strong were already described that way; for
case Core, the large set is all of $\Real^3$ and the nice behavior is
nonzero $\rho$. It is natural to consider whether some localization is
possible --- to what degree does nice behavior on a set smaller than $\Real^3$
imply some kind of uniqueness on the same set? This is a second major theme.
The cases in the Omnibus theorem are in fact corollaries of just such results: 
the local Core, Weak, and Strong HK Thms. \ref{thm:local Core HK}, 
\ref{thm:local Weak HK}, and \ref{thm:local Strong HK}. 
We do not give an account of those at this point, referring instead to the
precise restatement given in Section \ref{sec:summary},
together with a sharpening of the preceding Omnibus theorem.
Each case builds on the previous ones and each requires an additional
significant result to be called in from the literature. For Core,
it is that energy eigenfunctions are guaranteed to be continuous\cite{Aizenman+Simon} 
for potentials in our chosen universe. That is the main motivation for its choice.
For Weak, it is a weak unique continuation principle (UCP)\cite{Schechter+Simon},
and for Strong, it is a strong UCP\cite{Jerison+Kenig}.

Here is a brief guide to the organization of the paper.
Section \ref{sec:setting} gives some background, carefully 
carries out Step 1 referred to in the opening paragraph using the
constrained search principle\cite{Levy79}, and identifies a suitable universe
of potentials, $\Univ$, containing the Kato class $K_3$. Its main function
is to ensure that energy eigenfunctions are continuous and their densities 
lower semicontinuous. 
The real business of studying uniqueness of $\vext$ begins after that,
but in a generalized form, namely, given a wavefunction $\psi$ with density $\rho$, 
how many external potentials will make it an energy eigenstate?
This uniqueness question takes on a nuanced form with the consideration of
uniqueness sets smaller than $\Real^3$. Section \ref{sec:uniqueness}
sets forth the appropriate definitions and discussion.
Section \ref{sec:Core} then centers on case Core of the 
Omnibus HK theorem (no conditions on potential, i.e., over and above the
Kato class restriction).
Section \ref{sec:local-inversion} derives a
{\em constructive} local inversion of the Schr\"odinger equation,
determining $\vext$ from $\psi$ where the latter is nonzero, which is
then applied to obtain the local Core and Core HK theorems in 
Section \ref{sec:Core thms}.
Section \ref{sec:simple} presents the Simple HK theorem \ref{thm:Simple HK},
a simplified version of case Strong of the Omnibus theorem 
(no conditions on the density). In fact, it is precisely the simplified 
version quoted earlier. 
Readers satisfied with that and uninterested in pursuing matters 
further may wish to take that section as an exit ramp. 
Deeper exploration commences in Section \ref{sec:Weak}. A weak Unique Continuation 
Property (UCP) of Schechter \& Simon\cite{Schechter+Simon} is reviewed in 
Section \ref{sec:weak UCP}. Its application involves a property of potentials
identified in Section \ref{sec:Sobolev-multipliers} 
as being local Sobolev multipliers in the sense of Maz'ya\cite{Mazya-09-book}.
Section \ref{sec:Weak thms} then derives the Local Weak and Weak HK theorems, 
corresponding to case Weak of the Omnibus theorem.
The development of Section \ref{sec:Strong} parallels that of 
Section \ref{sec:Weak}. A strong UCP of Jerison \& Kenig\cite{Jerison+Kenig} 
is recalled in Section \ref{sec:strong UCP} and its applicability to
our problem established. 
The relevant property of potentials here is local $3\Num/2$-power integrability.
In Section \ref{sec:Strong thms}, the UCP is applied to
prove the Local Strong and Strong HK theorems.
The designations ``Core'', ``Weak'' and ``Strong'' are partially motivated by
the r\^ole that the weak and strong UCPs play in derivation of the latter two
categories of results, but this nomenclature is more mnemonic than systematic.
Finally, Section \ref{sec:summary} pulls the results together, giving a summary
similar to the Omnibus theorem above, but in a precise style and with 
sharpened conditions implying full uniqueness over $\Real^3$.

\section{Setting}
\label{sec:setting}

This Section prepares the ground to investigate the potential uniqueness problem
by carrying out Step 1 of the Introduction in Thm. \ref{thm:constrained search},
and determining a suitable universe of potentials.
Section \ref{sec:notation} establishes basic notation and Section \ref{sec:Kato}
gives background on Kato class and local Kato class potentials.
The constrained search principle is applied in Section \ref{sec:constrained search}
to reduce the original problem to one about representing potentials for
a given wavefunction. Section \ref{sec:forms} identifies a suitable universe
$\Univ$ as consisting of potentials with positive (negative) part in
$K_3$ ($K_{3,\tloc}$). 

\subsection{Notation}
\label{sec:notation}

Working in the standard formulation of quantum mechanics in 
Hilbert space\cite{von-Neumann,Prugovecki,Amrein-book,RSI,RSII},
we deal with a system of an arbitrary but fixed number, $\Num \ge 2$, of identical 
non-relativistic particles in three dimensions. (The case $\Num = 1$ is
exceptional in certain ways, but since it is not of much interest in the DFT 
context, we do not bother to point out where and how.)
The Hilbert space of pure states, $\Hilb$, is equipped with the inner product
\begin{equation}
\langle \psi | \phi \rangle = \sum_{\underline{\sigma}} \int_{\Real^{3\Num}} 
{\psi}({\underline{\sigma}};\underline{x})^*  
{\phi}({\underline{\sigma}};\underline{x}) \, 
d\underline{x}
\end{equation}
the sum being over all spin assignments 
$\underline{\sigma} \defeq (\sigma_1,\sigma_2,\ldots,\sigma_\Num)$
to the $\Num$ particles. An underlined symbol, as in this equation,
indicates a collection of things indexed from $1$ to $\Num$.
For physical applications, $\Hilb$ contains only states of appropriate 
exchange symmetry (fermionic or bosonic).
However, it turns out that neither exchange symmetry nor spin really 
plays any significant r\^ole in our considerations.
The Hamiltonian is built from one-body and two-body pieces:
\begin{equation}
H = T + \Vint + \Vext.
\label{eq:Hamiltonian}
\end{equation}
$T$ denotes kinetic energy, represented by the $3\Num$-dimensional Laplacian $-\nabla^2$. (Effectively, we take $\hbar^2/2m = 1$.)
For notational simplicity, we use the following notations for
potential energy of interaction among all the particles,
the external potential, and the total potential: 
\begin{equation}
\Vint(x_1,\ldots,x_\Num) = \sum_{1 \le i < j \le \Num} v_{\text{int}}(x_i-x_j),
\quad
\Vext (x_1,\ldots,x_N) = \sum_{i=1}^\Num \vext(x_i),
\quad\text{and}\quad
\Vtot = \Vint + \Vext.
\end{equation}
Associated to the wavefunction $\psi \in \Hilb_\Num$ is the (one-particle) density 
\begin{equation}
(\dens \psi)(x) = 
\sum_{1\le i\le \Num} 
\sum_{\underline{\sigma}} 
\int_{\Real^{3(\Num-1)}} |\psi({\underline{\sigma}};x_1,\ldots,x_{i-1},x,x_{i+1},\ldots,x_N)|^2 \, 
dx_1 \cdots dx_{i-1} dx_{i+1} \cdots dx_N.
\label{eq:density-defn}
\end{equation}
We use $\rho$ generically to denote the density associated with whatever state
is currently under discussion. 
Like wavefunctions, $\rho$ is not really a point function, being defined only up to 
modification on a measure zero set. This is not a pedantic point since we will
need a meaningful distinction between $\rho$ vanishing nowhere and vanishing on
a set of measure zero, and it will drive the selection of universe of potentials.
Recall in connection with these measure-theoretic issues that 
to say that $P$ holds almost everywhere in $X$ is the same as saying that
the subset of $X$ on which $P$ fails has zero measure.
Equivalently, the set on which $P$ holds is `of full measure' in $X$.
If $X$ is not explicitly mentioned in a statement of that form, it is
implied to be $\Real^n$, $n$ depending on context.
In the following, we frequently use the standard abbreviation
`a.e.' for `almost everywhere'. 

The dependence of $\dens\psi$ on $\psi$ is through a linear dependence on
the corresponding projection $|\psi\rangle\langle\psi |$. 
Therefore $\dens$ has an immediate extension by linearity to a mixed state
$\gamma = \sum_i \lambda_i \DiracProj{\psi_i}$ as
\begin{equation}
\dens \gamma = \sum_i \lambda_i \, \dens\psi_i.
\end{equation}

\subsection{Kato-class potentials}
\label{sec:Kato}

For a function $V$ on $\Real^n$, $n \ge 3$, define
\begin{equation}
M_V(x,\delta) = \int_{|x-y| \le \delta} \frac{|V(y)|\;\;\;}{|x-y|^{n - 2}}\, d^ny.
\label{eq:MV}
\end{equation}
If $M_V(x,\delta) \to 0$ as $\delta \to 0$, uniformly in $|x| < R$ 
for each $R<\infty$,
then $V$ is said to be in the {\it local Kato class} $K_{n,\tloc}$. 
If convergence to zero is uniform over all $|x|$, then $V$ is in
the {\it Kato class} $K_n$. Thus, $K_n \subset K_{n,\tloc}$. The condition
to be satisfied is fairly called ``local'' in either case, but for
$K_n$, it must be satisfied uniformly.
The Kato classes $K_{3\Num}$ will be important to this investigation,
but the following property allows us to reduce considerations to $K_3$.
If $V^\prime \in K_{n}$, $V^{\prime\prime} \in K_{m}$, and
${V}(x_1,\ldots,x_{n+m}) = 
V^\prime(x_1,\ldots,x_n) + V^{\prime\prime}(x_{n+1},\ldots,x_{n+m})$, then
${V} \in K_{n+m}$. A parallel relation holds for the local Kato classes.
Thus, for example, $\vext \in K_{3,\tloc}$ implies $\Vext \in K_{3\Num,\tloc}$.

To get a better feeling for the Kato classes, we compare
them to the Lebesgue and local Lebesgue scales. Recall that 
$L^p_\tloc(\Real^n)$ consists of functions $p$-th power 
integrable over any bounded subset of $\Real^n$.
Application of H\"older's inequality to $M_V(x,\delta)$ shows that
\begin{equation}
L^{p}_{\tloc}(\Real^3) \subseteq K_{3,\tloc} \,\;\text{and}\;\,
L^{p}(\Real^3) \subseteq K_{3}, 
\quad \text{for}\; p > \frac{3}{2}.
\label{eq:Kato vs. Lebesgue}
\end{equation}
In the other direction, by direct inspection of the integral in (\ref{eq:MV}),
\begin{equation}
K_{3,\tloc} \subseteq L^1_{\tloc}(\Real^3). 
\end{equation}
However, $K_3$ is not contained in $L^1(\Real^3)$. 
The condition to be in $K_3$ is local in a sense similar to that of the
local Lebesgue spaces. Thus, for example, singular
periodic potentials which may be of interest for solid-state physics
can be in $K_3$, whereas they are never in $L^1(\Real^3)$.
A potential which behaves as $V(x) \propto |x|^{-2} (\ln |x|)^{-\alpha}$ 
with $2/3 < \alpha \le 1$ for $x$ near $0$ and zero elsewhere
belongs to $L^{3/2}(\Real^3)$, but does not belong to $K_{3,\tloc}$, so
\begin{equation}
L^{3/2}(\Real^3) \nsubseteq K_{3,\tloc}.
\end{equation}
This observation is important insofar as it shows that neither 
the Lieb class 
$L^{3/2}(\Real^3) + L^\infty(\Real^3)$ is contained in $K_3$ nor vice versa.
But, as shown by (\ref{eq:Kato vs. Lebesgue}), $K_3$ nearly contains
the Lieb class in some sense.
For further information on Kato classes, see
Ref.~\onlinecite{Aizenman+Simon}, Ch. 2 of Ref.~\onlinecite{CFKS}, 
or \S A2 of Ref.~\onlinecite{Simon-82}.

\subsection{Constrained search principle}
\label{sec:constrained search}

The total energy of a normalized state $\psi$ in presence of the external 
one-body potential $\vext$ is
\begin{equation}
{\EE}_{\vext}[\psi] 
= {\EE}_0[\psi] + \langle \vext,\dens \psi \rangle,
\label{eq:psi-energy}
\end{equation}
where
\begin{equation}
{\EE}_0[\psi] = \langle \psi|T|\psi\rangle + \langle \psi|\Vint|\psi\rangle,
\;\; \text{and}\;\;
\langle \vext,\dens \psi \rangle = \int \vext\, (\dens \psi)\, dx.
\label{eq:energy forms}
\end{equation}
This extends from pure states $|\psi \rangle\langle \psi|$
to mixed states just as $\dens$ does.
A ground state is simply a (normalized) state minimizing the energy 
form (\ref{eq:psi-energy}).
The next theorem, embodying Step 1 of the Introduction,
is a form of the constrained search principle first emphasized by Levy\cite{Levy79}.
The idea requires no operator theoretic, or even Hilbert space, considerations; 
${\EE}_0$ and $\langle \vext,\, \cdot\, \rangle$ need only be functions 
into $\Real\cup\{+\infty\}$.
\begin{thm}
\label{thm:constrained search}
If $\rho$ is a ground density for both $\vext$ and $\vext^\prime$,
and $\gamma$ is a ground state for $\vext$ with $\dens\gamma = \rho$, 
then $\gamma$, along with every vector in its range space,
is a ground state also for $\vext^\prime$. 
\end{thm}
\begin{proof}
Among $\setof{\gamma}{\dens\gamma = \rho}$, all and only those states
which minimize $\EE_0$ are ground states for $\vext$, because
$\EE_{\vext}[\gamma] = \EE_0[\gamma] + \langle \vext,\rho\rangle$ 
is the same for all of them, and equal to the ground energy by assumption.
\end{proof}

\subsection{Quadratic forms, operators, and the universe of potentials}
\label{sec:forms}

Since we intend to deal with highly singular potentials, the issues of the domain 
of the energy function (\ref{eq:psi-energy}) and existence of a corresponding
self-adjoint Hamiltonian operator are nontrivial.
The kinetic energy functional is written above as
`$\Dbracket{\psi}{T}{\psi}$' instead of `$\inpr{\psi}{T\psi}$' because
it is intended to be understood, initially, as a quadratic form --- the
restriction to the diagonal of the sesquilinear form
\begin{equation}
\langle \phi | T | \psi \rangle \defeq 
\int \nabla\phi^* \cdot \nabla\psi \, d\underline{x}
= \int |\underline{p}|^2 
\tilde\phi(\underline{p})^* \tilde\psi(\underline{p}) \, d\underline{p},
\label{eq:T fnal}
\end{equation}
where tilde denotes Fourier transform.
Throughout this subsection, both spin and exchange statistics are ignored;
they contribute notational clutter but nothing substantial to the 
present considerations.
The point is that the sesquilinear form is well-defined for $\phi$ and $\psi$ in the
form domain $D[T]$ which is bigger than the operator domain of the
kinetic energy. Indeed, $D[T]$ is the Sobolev space
$H^1(\Real^{3\Num})$, a Hilbert space under the inner product
$\inpr{\phi}{\psi}_{H^1} = \Dbracket{\phi}{T+1}{\psi}$.
Similarly for a potential (e.g., $\Vext$ or $\Vint$),
\begin{equation}
\langle \phi | V | \psi \rangle \defeq 
\int \phi^* \psi\, V \, d\underline{x}.
\end{equation}
Suppose now that $C_c^\infty(\Real^{3\Num})$ (compactly supported, 
infinitely differentiable functions) is in the domain of the
real form 
${\cal E}_{\vext}[\psi] = \Dbracket{\psi}{T+\Vint+\Vext}{\psi}$.
Then, if ${\cal E}_{\vext}[\psi] = E_0$ is a minimum over normalized states, 
a simple integration by parts suffices to see that
\begin{equation}
0 = \inpr{(-\nabla^2 + \Vtot - E_0)\eta}{\psi} 
= \int\psi (-\nabla^2 + \Vtot - E_0)\eta^*\, d{\underline{x}},
\label{eq:Schrodinger 0}
\end{equation}
for every $\eta \in C_c^\infty(\Real^{3\Num})$.
Without loss of generality, we may take $E_0 = 0$ (with the
pleasant side effect of rendering normalization irrelevant).
Then, the previous display says precisely that $\psi$ is a solution 
in distribution sense of the Schr\"odinger equation
\begin{equation}
\label{eq:Schrodinger}
(-\nabla^2 + \Vint + \Vext) \psi = 0.
\end{equation}
To ensure that the form domain contains $C_c^\infty(\Real^{3\Num})$,
it suffices that $\vext$ and $\vint$ are locally integrable.
By a result of Aizenman and Simon
[Thm. 1.5 of Ref.~\onlinecite{Aizenman+Simon}], {\em all} distributional 
solutions of (\ref{eq:Schrodinger}) are actually continuous 
functions when $\vext$ and $\vint$ are in the local Kato class $K_{3,\tloc}$.
We will want such continuity for technical reasons, 
although one might also have reasons of a more philosophical nature.
At any rate, this motivates working in the universe
\begin{equation}
\Univ = K_{3,\tloc} \quad \text{(preliminary)}.
\label{eq:universe preliminary}
\end{equation}

On the surface, at least, the main part of this paper concerns 
uniqueness of $\vext$ in Eq. (\ref{eq:Schrodinger})
as it stands, that is without regard even to whether there are other
eigenfunctions for lower energy. However,
to be quantum mechanically respectable, our energy form should be not
only lower bounded, but also correspond in a natural way to a lower 
bounded self-adjoint operator. To be assured of that will require
additional restriction.
The idea is that if the quadratic form 
$\Dbracket{\,\cdot\,}{T+\Vint+\Vext}{\,\cdot\,}$ has
domain $D$ dense in $L^2(\Real^{3\Num})$, then 
a corresponding operator $H$ is unambiguously defined on a subspace of $D$ via
\begin{equation}
H\psi = \eta \in L^2(\Real^{3\Num})
\quad\Leftrightarrow\quad
\forall \phi\in D, \, \Dbracket{\phi}{T+\Vint+\Vext}{\psi} = \inpr{\phi}{\eta}.
\end{equation}
According to the standard 
theory,\cite{Kato,Simon-forms,RSI,RSII,deOliveira,Schmudgen}
lower bounded self-adjoint operators on a Hilbert space $\Hilb$
are in one-to-one correspondence with real lower bounded forms which 
are {\em closed} on a dense domain. 
What closedness means for such a form $\Dbracket{\psi}{A}{\psi}$ with 
domain $D[A]\subset \Hilb$ can be expressed in either of these two ways:
(i) for some real number $m$, $\Dbracket{\psi}{A+m}{\phi}$ is an inner
product making $D[A]$ a Hilbert space, (ii)
extending $\Dbracket{\psi}{A}{\psi}$ to $\Hilb$ by
declaring it equal to $+\infty$ off $D[A]$ produces a 
lower semicontinuous function on $\Hilb$.
(For this second, less common characterization, see \S 10.1 of Ref. \onlinecite{Schmudgen},
\S 9.3 of Ref. \onlinecite{deOliveira} or Ref. \onlinecite{Simon77}.)
It is easy to see that the kinetic energy functional (\ref{eq:T fnal})
is a closed, lower bounded, form on $H^1(\Real^{3\Num})$. 
To maintain those properties under addition of potentials requires 
differing considerations for the positive (repulsive) and negative (attractive)
parts. For the positive part, characterization (ii) shows that only local integrability
is required (to keep $C_c^\infty(\Real^{3\Num})$ in the domain). 
For the negative part, 
Kato shows (\S VI.4.6 of Ref. \onlinecite{Kato}), using characterization (i), 
that it is sufficient for it to be in the Kato class $K_3$. (There are other
sufficient conditions.)

These further considerations motivate reducing the universe of potentials to
\begin{equation}
\Univ = K_{3,\tloc}^+ - K_{3}^+.
\end{equation}
That is, the positive parts of potentials are still in $K_{3,\tloc}$ while 
negative parts are restricted to $K_{3}$. This restriction guarantees a good
quantum mechanical interpretation of the results of Sections 3 -- 6, but
those results do not themselves require the restriction.

\subsection{Spin components, 
exchange symmetry, and lower semicontinuity of the density} 
\label{sec:spin components}

To reinstate spin and exchange symmetry in the considerations of Section
\ref{sec:forms} merely requires adding spin components and restricting to
the correct symmetry Hilbert subspace, $\Hilb$.
But, if $\psi \in \Hilb$ satisfies (\ref{eq:Schrodinger 0}) 
for $\eta \in \Hilb$, then the symmetry restriction on $\eta$ is actually
dispensible. The conclusion is that each spin component of $\psi$ satisfies
the Schr\"odinger equation (\ref{eq:Schrodinger}) in its original, spinless, sense.
Thus, in the following sections, we work directly only with single-component
wavefunctions satisfying (\ref{eq:Schrodinger}) with no
exchange symmetry. It will be seen, as discussed in the next Section,
that information from distinct components can be easily patched together all the way to
mixed states. For a mixed ground state, each spin component of each vector
in its range space satisfies (\ref{eq:Schrodinger}).

Once the symmetry restriction is lifted, it becomes convenient to
use the partial densities
\begin{equation}
\rho_{n}(x) = 
\int_{\Real^{3(\Num-1)}} |\psi(x_1,\ldots,x_{n-1},x,x_{n+1},\ldots,x_N)|^2 \, 
dx_1 \cdots dx_{i-1} dx_{i+1} \cdots dx_N,
\label{eq:density-n-sigma}
\end{equation}
for $n=1,\ldots,\Num$, so that $\rho(x) = \sum_{n=1}^\Num \rho_n(x)$.
The choice of universe $\Univ$ was motivated by the fact that it 
makes an eigenfunction $\psi$ continuous. 
What that implies for the density $\rho$ or partial densities $\rho_n$
is lower semicontinuity.
Recall that a function $f$ is lower semicontinuous if $\setof{x}{f(x) \le c}$
is closed for every $c\in\Real$.
That continuity of $\psi$ implies lower semicontinuity of $\rho$ is seen
as follows.
With $\Lambda_R(x_i)$ a continuous cutoff function equal to 
1 for $|x_i| \le R$ and dropping monotonically to zero at $|x_i| = R+1$, 
$|\psi_R|^2 \defeq \prod_i \Lambda_R(x_i) |\psi(\underline{x})|^2$ is nonnegative,
continuous, and compactly supported (hence {\em uniformly} continuous).
Substituting into the formula (\ref{eq:density-defn}) yields a continuous
$\rho_R$. As $R\to\infty$, $\rho_R$ increases to $\dens\psi$ by the monotone 
convergence theorem. An increasing limit of continuous functions is 
lower semicontinuous, hence $\dens\psi$ is lower semicontinuous.
This argument uses the continuous version of $\psi$ to compute $\rho$.
We can also recover the lower semicontinuous version of $\rho$ directly from
any version as the function 
$x \mapsto \lim_{\delta \downarrow 0} \essinf_{B_\delta(x)} \, \rho$.
The importance of all this to our considerations is that lower semicontinuity of $\rho$
makes the notion of connected components of $\{\rho > 0\}$ well-defined. 

\section{Uniqueness}
\label{sec:uniqueness}

Henceforth, $\vint$ and $\vext$ are assumed to be in $\Univ$, but
they have slightly different status.
The interaction potential $\vint \in \Univ$, and the state,
whether a pure state ($\psi$) or a mixed state ($\gamma$), 
are considered as given, while $\vext$ is essentially a variable 
ranging over the set $\Potl \subseteq \Univ$ of potentials for which 
the Schr\"odinger equation (\ref{eq:Schrodinger}) is satisfied. 
Thus, an assertion of
the form, ``if $\vext$ has property $P$...'' is to be roughly understood
as existentially quantified and synonymous with ``if any $\vext\in\Potl$ has
property $P$...''. 
The set $\Potl$ is determined by the state and the interaction potential
via the Schr\"odinger equation, and the central concern is in what ways
its members can differ from one another, the most desirable case being
that $\Potl$ has a unique member (if any).
The following definition, formalizing modes of partial uniqueness,
will be of central importance, and it is formulated with an eye particularly
on the possibility that there may be disjoint sets $U$ and $U^\prime$ 
such that any pair of potentials in $\Potl$ differ by constants 
almost everywhere over each of $U$ and $U^\prime$, but that 
there are different choices for those constants.
\begin{defn}[uniqueness sets and points, $c$-equivalence]
\label{def:uniqueness}
An open set $U\subseteq \Real^3$ is an {\it uniqueness set} 
of $\Potl$ if, for each pair $\vext,\vext^\prime \in \Potl$,
$\vext - \vext^\prime$ is constant a.e. in $U$. 
The point $x$ is a {\it uniqueness point} if some open ball containing 
$x$ is a uniqueness set.
Two uniqueness sets $U$ and $U^\prime$ are $c$-{\it equivalent}, 
written $U\cequiv U^\prime$, if $U\cup U^\prime$ is a uniqueness set. 
Similarly, for two uniqueness points $x$ and $y$, $x\cequiv y$ means 
that $x$ and $y$ are in a common uniqueness set. 
\end{defn}
Often, when invoking this this definition, we will not bother to mention
$\Potl$, it being implicit in the discussion and determined through the
state and $\vint$ (the latter of which will also usually be implicit).

One might say that the subject of this paper is $c$-equivalence classes.
In that way of describing things, identifiable subrelations of $\cequiv$ will 
be very important. 
Given some such, the equivalence classes of
the generated equivalence relation are subsets of $c$-equivalence classes.
An important example of this simple principle is,
``have nonempty intersection'' is a subrelation 
of $\cequiv$, that is, if $U^{\prime}$ and $U^{\prime\prime}$ are uniqueness sets, 
then $U^{\prime}\cap U^{\prime\prime} \neq \varnothing$ implies
$U^{\prime}\cequiv U^{\prime\prime}$.
The reason is that, if $\vext^\prime - \vext$ is equal to constant $c^\prime$ 
or $c^{\prime\prime}$ almost everywhere on the open set $U^\prime$, respectively 
$U^{\prime\prime}$, then $\vext^\prime - \vext$ is equal to both 
$c^\prime$ and $c^{\prime\prime}$ almost everywhere on the nonempty open set
$U^{\prime}\cap U^{\prime\prime}$, implying $c^\prime = c^{\prime\prime}$. 
Here it was important that only {\it open} sets are eligible to be uniqueness
sets and if the intersection of two open sets is nonempty, it is open, hence of
nonzero measure.
The subrelation of $\cequiv$ just identified implies that the connected components 
of the set of all uniqueness points
are uniqueness sets; indeed, they are the equivalence classes generated by the
intersection subrelation.
This indicates that connectedness plays an important r\^ole in this paper, though 
it is needed only for open sets. Recall that an open set is connected precisely 
when it is not the union of two disjoint nonempty open sets. 
For a set $\Omega \subseteq \Real^n$, a connected component is a maximal connected subset,
and the notation
\begin{equation}
\comps \Omega = \text{set of connected components of }\Omega  
\end{equation}
will be convenient from time to time.
Note that, if $\Omega$ is open, every member of $\comps \Omega$ is also open.

Similarly, given a wavefunction satisfying the Schr\"odinger equation 
(\ref{eq:Schrodinger}), we can patch together the $c$-equivalence conclusions
resulting from consideration of individual spin components.
This leads to a style of working such that,
within a proof, $\rho$ implicitly corresponds to a
generic single-component wavefunction $\psi$ with no exchange symmetry
while in the statement of the same theorem, $\rho$ can refer to anything between
that level and a spin-full mixed state. Once the convention is understood, there
is little chance for confusion. We will generally refer to this procedure
simply as {\it patching}.

To see how Definition \ref{def:uniqueness} connects to our original problem, 
consider this template:
whenever $\dens\gamma$ has property $P$ and there is a $\vext\in\Potl$ with 
property $Q$, then $U$ is a uniqueness set. If this has been established,
Thm. \ref{thm:constrained search} immediately licenses the conclusion that, 
whenever $\rho$ has property $P$ and one representing potential has property 
$Q$, then every pair of representing potentials differ merely by a constant 
a.e. on $U$.
A conclusion of the traditional Hohenberg-Kohn form would say that $\Real^3$ is
a uniqueness set. In that case, in fact, $\Potl$ reduces to a singleton because
(\ref{eq:Schrodinger}) eliminates the freedom of a global constant shift by
specifying that the eigenvalue is zero.
Actually, the traditional Hohenberg-Kohn form would not only specify
$U = \Real^3$, but would omit condition $P$ in the above template.
However, there is nothing about DFT which compels such a narrow attitude.
It seems perfectly reasonable to ask what kinds of trade-offs can be 
made to weaken $Q$ by strengthening $P$. This is one theme
of the investigation. A second is the consideration of uniqueness sets which are
not all of $\Real^3$. One motivation for considering such results is to
find out to what degree ``good'' behavior of $\rho$ and/or $\vext$ can imply
uniqueness {\em locally}. 

\section{Core} 
\label{sec:Core}

This Section presents the easiest results in the direction of
what was called Step 2 in the Introduction.
Section \ref{sec:local-inversion} gives a construction of $\vext$
satisfying the Schr\"odinger equation (\ref{eq:Schrodinger}) locally 
in a set where $\psi$ is almost everywhere nonzero, up to an overall constant. 
The construction is modeled on the heuristic ``divide by $\psi$'' strategy
referenced in the Introduction.
It is used to show (Thm. \ref{thm:local Core HK}) that each connected 
component of $\{\rho > 0\}$ is a uniqueness set.
With the condition that $\rho > 0$ everywhere, our first HK theorem follows.
Finally, Section \ref{sec:from one to another} opens the discussion of
how to join connected components into a single uniqueness set and provides
a technical tool which will be used for that purpose later.
The adjective ``Core'' is meant to highlight that the
results here are fundamental and will be built upon in the following Sections.

\subsection{Local inversion of the stationary Schr\"odinger equation}
\label{sec:local-inversion}

Suppose that $\psi$ is a solution of the Schr\"odinger equation (\ref{eq:Schrodinger}),
nonzero almost everywhere on $B_r(y_1) \times B_r(y_2) \times \cdots \times B_r(y_\Num)$,
where $B_r(x)$ denotes the open ball of radius $r$ centered at $x$.
We locally invert the equation to obtain $\vext$ on $B_r(y_1)$, 
up to a constant.
Rearrange the equation to [it also works to hold division by $\psi$ in abeyance until 
the last step] 
\begin{equation}
\vext(x_1) = \frac{T\psi}{\psi}(\underline{x}) 
- \sum_{1 \le  i < j \le \Num} \vint(x_i-x_j) 
- \sum_{2 \le i \le \Num} \vext(x_i).
\label{eq:solved-for-v}
\end{equation}
The idea is simply to freeze the last $\Num-1$ coordinates to
$(y_2,\ldots,y_\Num)$ and vary $x_1$ in $B_r({y_1})$. 
Since the first two terms on the right-hand side of (\ref{eq:solved-for-v}) 
are known and the last is constant, $\vext(x_1)$ can be extracted.
The difficulty is that the equation holds only almost everywhere,
while the slice $(x_2,\ldots,x_\Num) = (y_2,\ldots,y_\Num)$ has measure zero. 
To cope with that, smear everything with the aid of
a function $h \in C^\infty_c(\Real^3)^+$ --- 
smooth, non-negative, supported in the unit ball $B_1(\Real^3)$, and having integral $1$.
The scaled version
\begin{equation}
h_\epsilon(x) = \epsilon^{3} h\left(\frac{x}{\epsilon}\right)
\end{equation}
is supported in $B_\epsilon(\Real^3)$ and also has integral $1$.
Convolution of (\ref{eq:solved-for-v}) with $h_\epsilon^\Num$ yields
\begin{equation}
\vext^\epsilon(x_1) = 
\left[ h_\epsilon^\Num * \frac{T\psi}{\psi} \right](x_1,y_2,\ldots,y_\Num)
- \sum_{2 \le j \le \Num} \vint^\epsilon(x_1-y_j) 
- \sum_{2\le i < j \le \Num} \vint^\epsilon(y_i-y_j) 
- \sum_{2 \le i \le \Num} \vext^\epsilon(y_i) 
\label{eq:solved-for-v-and-smeared}
\end{equation}
where
\begin{equation}
  \vext^\epsilon(x) = (h_\epsilon * \vext)(x)
 = \int_{\Real^3} h_\epsilon(y) \vext(x-y)\, dy, 
\quad
  \vint^\epsilon(x) 
 = \int_{\Real^6} h_\epsilon(y) h_\epsilon(y^\prime) \vint(x-y+y^\prime) \, dy\, dy^\prime, 
\end{equation}
Eq. \ref{eq:solved-for-v-and-smeared} is well-defined for 
$|x_1 - y_1| < r-\epsilon$.
Every term is a smooth function of $x_1 \in B_{r-\epsilon}(y_1)$, 
the first two terms on the right-hand side are known and the last two are
constants, so $\vext^\epsilon$ is determined over $B_{r-\epsilon}(y_1)$, up to a constant.
As $\epsilon \to 0$, $\vext^\epsilon$ converges to $\vext$ in $L^1(B_{r-\delta}(y_1))$ 
for any fixed $\delta$. This determines $\vext$ over $B_{r}(y_1)$, up to a constant.

\subsection{Core HK theorems}
\label{sec:Core thms}

\begin{thm}[local Core HK]
\label{thm:local Core HK}
Each connected component of $\{\rho > 0\}$ is a uniqueness set.
\end{thm}
\begin{proof}
Suppose $\rho_1(x_1) > 0$. Since $\psi$ is continuous,
there must be some $r > 0$ and $x_2,\ldots,x_\Num$ 
such that $\psi \neq 0$ on 
$B_r(x_1)\times B_r(x_{2})\times \cdots\times B_r(x_{\Num})$.
Using the local inversion procedure of Section \ref{sec:local-inversion}, 
$\vext$ can be determined on $B_r(x_1)$ up to a constant. 
Now apply {\it patching}.
\end{proof}

Obtaining a stronger conclusion is the major preoccupation of the rest of
the paper. There are two ways to do that: hypotheses on the potentials $\vext$
and $\vint$ or hypotheses on $\rho$. One could use hypotheses on $\psi$, but
for DFT purposes that is inappropriate. 
In that context, we can suppose information about $\rho$ is available, 
but not about $\psi$, except what is implied by the density.
So, we ask what conditions on $\rho$ would close the
gap between the Core HK thm. \ref{thm:local Core HK} and the traditional statement.
The answer is immediately forthcoming:
if $\rho > 0$ everywhere, then $\Real^3$ itself is the unique
connected component of $\{\rho > 0\}$.
\begin{thm}[Core HK]
\label{thm:Core HK}
If $\rho > 0$ everywhere, then $\Real^3$ is a uniqueness set.  
\end{thm}

\subsection{From one connected component to another}
\label{sec:from one to another}

With the notations $\cl X$, $\intr X$ and $\bndy X$ for
the closure, interior and boundary, respectively, of a set $X$,
$\Real^3$ can be decomposed in these alternate ways as the 
union of two (jointly dense) open sets and their common boundary:
\begin{align}
\Real^3 
& = 
\intr \{\rho = 0\} 
\cup
\{\rho > 0\}
\cup 
\bndy \{\rho>0\} 
\nonumber \\
& = 
\intr \{\rho=0\} 
\cup 
\intr \cl \{\rho > 0\} 
\cup 
\bndy \intr \{\rho=0\}.
\end{align}
The Core HK thm. \ref{thm:local Core HK} shows that $\{\rho > 0\}$ 
consists of uniqueness points.
On the other hand, if $x \not\in \intr\cl \{\rho > 0\}$,
then $x \in \cl \intr \{\rho = 0\}$. Every neighborhood of $x$ contains
an open subset of $\{\rho = 0\}$, where $\vext$ is entirely unconstrained,
so $x$ cannot possibly be a uniqueness point.
Summing up, 
\begin{equation}
\{\rho>0\} \subseteq \{\text{uniqueness points}\} 
\subseteq \intr \cl \{\rho>0\}.
\label{eq:inclusions}
\end{equation}
Thus, for there to be any hope that $\Real^3$ is a uniqueness set,
$\{\rho > 0\}$ must be dense in $\Real^3$. And when the latter has
multiple connected components, a way to show their $c$-equivalence
is needed. The following lemma points a way toward that.
\begin{lem}
\label{lem:single particle transfer}
Let $U_1,U_1^\prime,U_2,\ldots,U_\Num \in \comps \{\rho > 0\}$.
If both $U_1\times U_2\times \cdots\times U_\Num$
and $U_1^\prime\times U_2\times \cdots\times U_\Num$ intersect $\{\psi\neq 0\}$,
then $U_1 \cequiv U_1^\prime$.
\end{lem}
\begin{proof}
Fix $\vext, \vext^\prime \in \Potl$.
According to Thm. \ref{thm:local Core HK}, 
$\vext^\prime - \vext$ is almost everywhere equal to some constant $c(U)$
over $U \in \comps \{\rho>0\}$.
Taking the difference of the Schr\"odinger equations corresponding
to $\vext^\prime$ and $\vext$ somewhere in 
$(U_1\times U_2\times \cdots\times U_\Num)\cap \{\psi\neq 0\}$
yields $c(U_1) + \sum_{n=2}^\Num c({U_n}) = 0$. Similarly,
$c(U_1^\prime) + \sum_{n=2}^\Num c({U_n}) = 0$. Hence, $c(U_1) = c(U_1^\prime)$.
\end{proof}
Effective use of this lemma requires finding appropriate conditions
on $\vint$ and $\vext$.

\section{Simple}
\label{sec:simple}

\subsection{Unique continuation properties}
\label{sec:ucp preview}

To progress beyond the Core HK theorems, 
we call on a powerful class of results known as unique continuation
properties (UCPs).
Suppose $\psi$ is a solution to the Schr\"odinger equation 
(\ref{eq:Schrodinger}) on the connected
open set $\Omega \subseteq \Real^{3\Num}$.
For our purposes, a UCP in that setting is a theorem of the following form:
If $\Vtot$ is in some class ${\cal C}$, and 
$\psi$ does not vanish identically over $\Omega$, then it cannot 
``vanish nontrivially''.
For a {\it weak} UCP, ``vanish nontrivially'' means ``vanish on an open set''.
For a {\it measure} UCP (there is really no standard terminology for this
case), ``vanish nontrivially'' means ``vanish on a set of nonzero measure''.
A third sense of ``vanish nontrivially'' will be met in Section \ref{sec:strong UCP}.
If we aim for a result like the usual Hohenberg-Kohn theorem, a measure
UCP is appropriate. The conclusion of a measure UCP is stronger than that of
a weak UCP, hence will require stronger hypotheses (i.e., a smaller class
${\cal C}$).
The major shortcoming for our purposes is that, while a UCP can guarantee
that $\{\psi \neq 0\}$ is very large, it will not guarantee that it is
connected.

In Thm. \ref{thm:Simple HK} just below, we use a measure UCP 
(Cor. \ref{cor:measure UCP})
for which ${\cal C}$ is the class of functions having locally integrable $3\Num/2$ power.
That is, $\Vtot \in {\cal C}$ if $|\Vtot|^{3\Num/2}$ is integrable over some 
neighborhood of every point in $\Omega$, a situation denoted 
$\Vtot \in L_{\tloc}^{3\Num/2}(\Omega)$.
This measure UCP follows from a result of Jerison \& Kenig~\cite{Jerison+Kenig}, 
the conclusion of which is stated as Thm. \ref{thm:Jerison+Kenig} below.
Requiring $\Vtot$ to be in $L_{\tloc}^{3\Num/2}(\Omega)$ allows only extremely weak 
singularities {\em within} $\Omega$ if $\Num$ is large.
However, we do not need $\Omega$ to be all of $\Real^3$, but only of
full measure.

\subsection{Simple HK theorem}

\begin{thm}[Simple HK]
\label{thm:Simple HK}
Suppose $\vint \in L_{\tloc}^{3\Num/2}(\Real^3\setminus\{0\})$ and
$\vext \in L_{\tloc}^{3\Num/2}(\Omega)$ for some open connected set $\Omega$ of full
measure in $\Real^3$ (i.e., $\Real^3\setminus \Omega$ has measure zero). 
Then, $\Real^3$ is a uniqueness set.
\end{thm}
\begin{proof}
Clearly, $\Omega^\Num \defeq \Omega\times \Omega\times \cdots \times \Omega$ 
is a connected open set of full measure in $\Real^{3\Num}$, 
and $\Vext \in L_{\tloc}^{3\Num/2}(\Omega^\Num)$.
Similarly, $\Vint$ is locally $3\Num/2$-integrable
away from points corresponding to coincidence of two or more particles. 
The set $\{x_i = x_j\}$ of coincidence of particles $i$ and $j$ is a linear subspace of 
codimension 3, and the result of removing it from a connected open set of full 
measure still has all three of those attributes:
Openness and being of full measure are clear. Connectedness follows since the
coincidence set has codimension greater than 1.
Therefore, $\Vtot$ is locally $3\Num/2$-integrable on some open, connected,
full-measure set $U \subseteq \Omega^\Num \subseteq \Real^{3\Num}$.
This sets the stage for the application (to $U$) of the measure UCP 
previewed in Section \ref{sec:ucp preview}.
Since $\psi$ cannot be identically zero, $\{\psi\neq 0\}$ is a full-measure 
subset of $U$, hence of $\Real^{3\Num}$, and $\{\rho > 0\}$ a full measure 
subset of $\Real^3$.
With these conclusions, the proof is now completed by an easy application of 
Lemma \ref{lem:single particle transfer}: $\{\psi\neq 0\}$ must
intersect any open set, hence $\{\rho > 0\}$ is a uniqueness set, and
then $\Real^3$ is also. 
\end{proof}

Much further along in this paper, the Strong HK Thm. \ref{cor:Strong HK}
will provide a strengthening of the Simple HK theorem above 
(same conclusion under strictly weaker hypotheses). 
For instance, the full-measure sets on which $\vext$ and $\vint$ are 
locally $3\Num/2$-integrable are not required to be connected. 
However, that result should not be viewed as the singular culmination 
toward which the development drives. 
Rather, the goal is to understand the trade-offs between conditions 
on potentials and on densities which will secure uniqueness, and
the degree to which those considerations can be localized.

\section{Weak}
\label{sec:Weak}

In order to improve on the local Core HK Thm. \ref{thm:local Core HK}, 
uniqueness must be extended from one connected component of 
$\{\rho > 0\}$ to another, and
Lemma \ref{lem:single particle transfer} offers a way to do that.
In this Section, conditions on $\vint$ and $\vext$ are found
which will allow the application of that lemma.
If they are both ``nice enough'', then we will be able to deduce
$x\cequiv y$ for two points of $\{\rho > 0\}$ if they are in the
same connected component of the ``nice'' set $\Md(\vext)$.
The key to the development is a weak UCP of Schechter \& Simon
discussed in Section \ref{sec:weak UCP}. 
Section \ref{sec:Sobolev-multipliers} identifies the 
property for applicability of that UCP to be that
the potentials are locally Sobolev multipliers, mapping $H_0^1$ into $L^2$.
For effective application, leading to the local Weak HK Thm. \ref{thm:local Weak HK}
and Weak HK Thm. \ref{thm:Weak HK},
$\vext$ will be required to have this property on a dense set, 
and $\vint$ on a set which, in addition, has a strong type of connectivity. 

In connection with these remarks, we recall that the topological notion
of denseness is very different from the measure-theoretic notion of 
almost everywhere. 
The rationals are a dense closed subset of $\Real$ with zero measure,
for instance. An example of a dense {\em open} subset with measure
as small as desired is provided by the complement of a fat Cantor set:
Beginning with the unit cube, remove open balls of radius $\epsilon$
centered at all points with coordinates $x_i \in \Ints$, then of radius
$\epsilon/2^\alpha$ centered at $x_i \in \Ints/2$, ..., of radius
$\epsilon/2^{n\alpha}$ centered at $x_i \in \Ints/2^n$, ....
After all this removal, the remaining (fat Cantor) set 
is a closed set with empty interior. 
However, the measure of the removed set does not exceed 
$c \epsilon\sum_{n\ge 0} 2^{(d-\alpha)n}$; for $\alpha > d$, this can
be made as small as desired by taking $\epsilon$ small. 

\subsection{Schechter-Simon weak UCP}
\label{sec:weak UCP}

Our next restriction on potentials is motivated by the possibility
of using the weak unique continuation property (UCP) in Thm.~\ref{thm:Schechter+Simon}.
We recall that for an open set $\Omega$, 
the Sobolev space~\cite{Adams,Lieb+Loss} 
$H^1_0(\Omega)$ is the Hilbert space obtained by 
completing $C_c^\infty(\Omega)$ with respect to (the norm derived from) the 
inner product
\begin{equation}
\langle f | g \rangle_{H^{1}(\Omega)} = \int_\Omega \overline{f}g + 
\nabla\overline{f}\cdot\nabla g\, dx.
\end{equation}
\begin{thm}[Schechter \& Simon, Thm. 2.1 of Ref.~\onlinecite{Schechter+Simon}]
\label{thm:Schechter+Simon}
Suppose that for some open ball $B \subseteq\Real^n$,
$V:B \rightarrow \Real$, and constant $C$,
\begin{equation}
\|V\eta \|_{L^2(B)} \le C \|\eta \|_{H^1(B)}
\label{eq:schechter-simon-condition}
\end{equation}
for all $\eta \in C_c^\infty(B)$.
Then, if $u$ (not presumed to be smooth) satisfies
\begin{equation}
|\nabla^2u| \le |Vu|,
\label{eq:diff-inequality}
\end{equation}
$u$ cannot vanish on an open subset of $B$ without vanishing on all of it.
\end{thm}

Note that the conclusion extends immediately to any connected open set
$\Omega = \bigcup_\alpha B_\alpha$, where for each open ball $B_\alpha$
there is a corresponding $C_\alpha$ as in the theorem.
For, if $I = \setof{\alpha}{u \;\text{vanishes on an open subset of}\; B_\alpha}$,
$\Omega^\prime = \bigcup_{\alpha\in I}B_\alpha$ and 
$\Omega^{\prime\prime} = \bigcup_{\alpha\not\in I}B_\alpha$, 
then $\Omega = \Omega^\prime \cup \Omega^{\prime\prime}$. However, according
to Thm. \ref{thm:Schechter+Simon}, $u \equiv 0$ on $\Omega^{\prime}$,
so that $\Omega^\prime$ and $\Omega^{\prime\prime}$ are disjoint.
Since $\Omega$ is connected, one of $\Omega^\prime$ and $\Omega^{\prime\prime}$ 
must be empty.

\subsection{Spaces of Sobolev multipliers}
\label{sec:Sobolev-multipliers}

This section formalizes and elucidates the condition on $V$ identified 
in Thm.~\ref{thm:Schechter+Simon}.
Since $C_c^\infty(\Omega)$ is dense in $H^1_0(\Omega)$, 
that condition is rephrased as: multiplication by $V$ is a bounded operator 
$H^1_0(\Omega) \rightarrow L^2(\Omega)$ with norm less than or equal to ${C}$.

\begin{defn}[Sobolev multipliers\cite{Mazya-85-book,Mazya-09-book}]
\label{def:sobolev-multiplier}
For an open connected set $\Omega \subseteq \Real^d$, 
\hbox{$M(H^1_0(\Omega) \rightarrow L^2(\Omega))$} denotes the
set of functions $f: \Omega\rightarrow \Real$ 
such that $fg \in L^2(\Omega)$ whenever $g\in H^1_0(\Omega)$.
Multiplication by $f \in M(H^1_0(\Omega) \rightarrow L^2(\Omega))$ 
is automatically bounded as a linear operator 
$H^1_0(\Omega) \rightarrow L^2(\Omega)$ 
by an argument based on the Closed Graph theorem, so the operator norm
\begin{equation}
\|f\|_{M(H^{1}_0(\Omega) \rightarrow L^{2}(\Omega))} = \sup 
\setof{ \|fg\|_{L^{2}(\Omega)} }{\|g\|_{H^{1}_0(\Omega)} = 1}
\end{equation}
is well-defined.
\end{defn}

For any open ball $B \subseteq \Real^3$, 
$L^3(B) \subset M(H^{1}_0(B) \rightarrow L^{2}(B)$
follows from a H\"older, then a Sobolev, inequality [for the latter, see 
e.g., Thm. V.5.4 of Ref. \onlinecite{Adams}]:
\begin{equation}
\|Vf\|_{L^2} \le \|V\|_{L^3} \|f\|_{L^6} \le c \|V\|_{L^3} \|f\|_{H^{1}}.  
\end{equation}
We will show that the larger space $L^3_{\text{wk}}(B) \supset L^3(B)$
is also contained in the multiplier space $M(H^{1}_0(B) \rightarrow L^{2}(B))$.
A function $f$ is in the weak $L^p$ space $L^p_{\text{wk}}(\Omega) \supset L^p(\Omega)$ 
if the quasinorm~\cite{Grafakos-Classical-Fourier-Analysis,Lieb+Loss}
\begin{equation}
\|f\|_{p,\mathrm{wk}}^p \defeq \sup_t\, t^{p-1}\Phi_f(t),
\end{equation}
is finite, where (`$\Leb$' denotes Lebesgue measure)
\begin{equation}
\Phi_f(t) \defeq \Leb{\setof{x\in\Omega}{ |f(x)| > t}}.
\end{equation}

Combining Eq. (2.3.4) and Prop. 1.2.8 of Ref. \onlinecite{Mazya-09-book},
[equivalently Eq. (3), \S 1.3.2 and Prop. 1, \S 1.1.4 of 
Ref. \onlinecite{Mazya-85-book}] yields
\begin{equation}
\|f \|_{M(H^1(\Real^3)\rightarrow L^2(\Real^3))}^2 
\le c
\sup_{\mathrm{diam}\, A \le 1} \frac{\|f \chi_A \|_{L^2}^2}{(\Leb{A})^{1/3}}.
\end{equation}
Clearly, among sets $A$ of given measure, 
the ratio inside the final $\sup$ is maximized by one of the form $A = \{|f| > t\}$.
Therefore,
\begin{equation}
\label{eq:ratio-1}
\|f \|_{M(H^1(\Real^3)\rightarrow L^2(\Real^3))}^2 
\le c
\sup_{t} \frac{\int_0^{\Phi_f(t)} s^2 d[\Phi_f(s)]}{(\Phi_f(t))^{1/3}},
\end{equation}
Rewriting the Stieltjes integral in the numerator as
\begin{equation}
\int_0^{\Phi_f(t)} s^2 d[\Phi_f(s)] = 
3\int_0^{\Phi_f(t)^{1/3}} [s^3 \Phi_f(s)]^{2/3} d[\Phi_f(s)^{1/3}]
\le 3{\Phi_f(t)^{1/3}} [\sup_s s^3 \Phi_f(s)]^{2/3}
\end{equation}
we see that the ratio (\ref{eq:ratio-1}) is bounded if $f$ belongs to $L^3_{\mathrm{wk}}$.
That is,
$L^3_{\mathrm{wk}}(\Real^3) \subseteq M(H^{1}(\Real^3) \rightarrow L^{2}(\Real^3))$.
Finally, noting that $L^3_{\text{wk}}(B)$ is isometrically embedded in 
$L^3_{\text{wk}}(\Real^3)$, while there is a bounded extension operator
$H^1_0(B) \rightarrow H^1(\Real^3)$, we obtain
\begin{equation}
L^3_{\mathrm{wk}}(B) \subseteq M(H_0^{1}(B) \rightarrow L^{2}(B)).
\end{equation}
The singularity of the Coulomb potential is just weak enough to be in
$L^3_{\mathrm{wk}}(\Real^3)$. The potential of a dense set of point charges in $\Real^3$
is then in $L^3_{\mathrm{wk,loc}}(\Real^3)$ as long as for every bounded set,
the sum of absolute values of charges therein is finite.

The following lemma is important because we will need to know that 
$\Vext$ inherits from  $\vext$ the property of being a Sobolev multiplier. 
\begin{lem}
If $f \in M(H^1_0(\Omega) \rightarrow L^2(\Omega))$ 
then 
\newline
(i)
$f \in M(H^1_0(\Omega^\prime) \rightarrow L^2(\Omega^\prime))$, 
where $\Omega^\prime \subseteq \Omega$;
(ii) $f\circ\pi^n_m \in M(H^1_0(\Omega\times\Real^{n-m}) \rightarrow L^2(\Omega\times \Real^{n-m}))$
where $\Omega \subseteq \Real^m$, and $\pi^n_m:\Real^n \rightarrow \Real^m$ is 
projection on the first $m$ components.
\end{lem}
\begin{proof}
(i) is immediate. (ii) follows from the fact that 
$H^1_0(\Omega\times\Real^{n-m})$ and $L^2(\Omega\times\Real^{n-m})$ are tensor products.
\end{proof}

\subsection{Multiplier sets and super-connectedness}
\label{sec:multiplier sets}

\begin{defn}[Multiplier points and sets]
\label{def:multiplier}
$x$ is a {\it multiplier point\/} of the potential 
$v:\Real^n \rightarrow \Real$ 
($n=3$ for $\vext$ or $\vint$, $n = {3\Num}$ for $\Vext$, $\Vint$ or $\Vtot$)
if there is an open connected neighborhood 
$U$ of $x$ such that the restriction of 
$v$ to $U$ is in $M(H^{1}_0(U) \rightarrow L^2(U))$.
The {\it multiplier set} of $v$, denoted $\Md(v)$, is the
set of all multiplier points of $v$, and is necessarily open. 
\end{defn}

Effective utilization of the Schechter-Simon UCP requires that
$\Md(\Vtot)$ be dense in $\Real^{3\Num}$.
That is because $\rho_1(x_1)$, say, is an integral over
$\Num-1$ particle coordinates and therefore reflects
the potential at points far distant from $x_1$. This nonlocality
was already used to advantage in the Core HK thm. \ref{thm:Core HK}, 
and will be advantageous
again in the next section. Here, however, it is mostly constraining.
The Schechter-Simon UCP \ref{thm:Schechter+Simon} implies that
if $\psi$ vanishes on an open subset of $U\in \comps\Md(\Vtot)$, 
then it vanishes identically on $U$. In general, however, this
is not useful for deducing anything about $\rho$.
For example, it is not ruled out that all of the density in $U$ is 
inherited from the region outside $\Md(\Vtot)$.

Not only denseness, but also connectivity, matters.
The connected components of $\Md(\Vext)$ are simply cartesian products
of those of $\Md(\vext)$. However, even if $\Md(\Vint)$ is connected and dense,
some of those components of $\Md(\Vext)$ could become chopped up in
$\Md(\Vtot)$. Hence, we will require something stronger than connectivity
of $\Md(\vint)$ in order to make progress.
\begin{defn}[Super-connected] A set $X$ is {\it super-connected} if,
for any open connected set $U$, $X\cap U$ is connected.
\end{defn}
We need this concept only for open dense sets, the
crucial point --- that $\Md(\Vint)$ is open dense and super-connected (ODS)
whenever $\Md(\vint)$ is so ---
following from the observations that being ODS 
is a property stable under both intersection and cartesian product. 
The first of these is quite simple: let $A$ and $B$ be ODS and $U$ open and connected.
Then, $(A\cap B)\cap U = A\cap (B\cap U)$ is open and connected.
For stability under cartesian product, we use the next Lemma, which shows
that an open dense set is super-connected if it has connected intersections with 
all elements of a base for the topology, such as 
open balls of radius less than some $r$, or open rectangles in $\Real^{3\Num}$.
For, if $A\subseteq \Real^n$ and $B\subseteq \Real^m$ are ODS, 
a base of open rectangles in $\Real^{n+m}$ witnesses that $A\times \Real^m$ and
$\Real^n\times B$ are also, and intersection-stability then shows that
$A\times B$ is ODS. 
\begin{lem}
If open dense set $A$ has connected intersection with all members of a
base for the topology, then $A$ is super-connected.  
\end{lem}
\begin{proof}
For an open connected set $U$, write $U=\bigcup {\cal C}$, 
where ${\cal C}$ is a subset of the base.
If $A\cap U = U^\prime \cup U^{\prime\prime}$ for disjoint open sets 
$U^\prime$ and $U^{\prime\prime}$, let 
${\cal C}^\prime = \setof{W \in {\cal C}}{A\cap W \subseteq U^\prime}$ and 
define ${\cal C}^{\prime\prime}$ similarly from $U^{\prime\prime}$.
Now, $\bigcup {\cal C}^{\prime}$ cannot be disjoint from $\bigcup {\cal C}^{\prime\prime}$ 
since $U$ is connected, hence there are 
$W^\prime \in {\cal C}^{\prime}$ and $W^{\prime\prime}\in {\cal C}^{\prime\prime}$ with
$A\cap W^\prime \cap W^{\prime\prime} \neq \varnothing$, since $A$ is dense.
But this contradicts 
$A\cap W^\prime \cap W^{\prime\prime} \subseteq U^\prime \cap U^{\prime\prime} = \varnothing$.
\end{proof}

\subsection{Weak HK theorems}
\label{sec:Weak thms}

Now we want to consider the following hypotheses:
\begin{align}
& \Md(\vint) \;\text{is dense in}\; \Real^3\;\text{and super-connected}
\label{eq:H1}
\tag{H1}
\\
& \Md(\vext) \;\text{is dense in}\; \Real^3
\label{eq:H2}
\tag{H2}
\\
& \Md(\vext) \;\text{is connected}
\label{eq:H3}
\tag{H3}
  \end{align}
The keys to effective use of these hypotheses are the
Schechter-Simon UCP \ref{thm:Schechter+Simon} and 
Lemma \ref{lem:single particle transfer}.
The development in this subsection is split into two tracks.
Section \ref{sec:H1+H2} works under assumption of only (H1) and (H2).
The results of Section \ref{sec:H1+H2+H3}, assuming also (H3), 
could be derived quickly as corollaries, but it is convenient to
give an independent, and simpler, development. The reader can proceed
to Section \ref{sec:Strong} after Section \ref{sec:H1+H2+H3}.
However, the local Weak HK Thm. \ref{thm:local Weak HK} of Section \ref{sec:H1+H2}
becomes relevant again in the Summary Section \ref{sec:summary}.

\subsubsection{consequences of (H1), (H2), and (H3)}
\label{sec:H1+H2+H3}

\begin{lem}
\label{lem:H123}
Assume (\ref{eq:H1}), (\ref{eq:H2}), and (\ref{eq:H3}).
Then, $\{\psi \neq 0\}$ is dense in $\Real^{3\Num}$.
\end{lem}
\begin{proof}
$\Omega = \Md(\vext)$ is a connected set dense in $\Real^3$,
hence $\Omega^\Num$ is connected and dense in $\Real^{3\Num}$. 
Hypothesis (H1) then ensures that $\Md(\Vtot)$ is also dense and connected.
Now, if $\{\psi= 0\}$ had an interior, it would intersect $\Md(\Vtot)$, 
the Schechter-Simon UCP would imply that $\psi = 0$ everywhere on $\Md(\Vtot)$,
and denseness of the latter would imply that the continuous function $\psi$ was
identically zero. That being impossible, conclude that $\{\psi \neq 0\}$ is
dense.
\end{proof}

\begin{lem}
\label{lem:Weak 2}
Assume (\ref{eq:H1}), (\ref{eq:H2}), and (\ref{eq:H3}).
Then, 
$\{\rho > 0\}$ is a uniqueness set dense in $\Real^3$. 
Therefore, all uniqueness points are $c$-equivalent.
\end{lem}
\begin{proof}
Lemma \ref{lem:H123} immediately implies that
$\{\rho > 0\}$ must be dense. But, it also provides conditions to
apply Lemma \ref{lem:single particle transfer} since now for any
connected components $U_1,\ldots,U_\Num$ of $\comps \{\rho > 0\}$,
$U_1\times \cdots \times U_\Num$ is guaranteed to intersect $\{\psi\neq 0\}$.
As a result, we can conclude that $\{\rho > 0\}$ is itself a uniqueness set. 
That is the first conclusion. Since $\{\rho > 0\}$ is dense, every uniqueness set
must intersect it, and the second conclusion follows by patching.
\end{proof}
\begin{thm}[Weak HK]
\label{thm:Weak HK}  
Assume (\ref{eq:H1}), (\ref{eq:H2}), and (\ref{eq:H3}).
Then $\Real^3$ is a uniqueness set if and only if $\{\rho= 0\}$ has zero measure.
\end{thm}
\begin{proof}
$\{\rho=0\}$ having zero measure is certainly a necessary condition for
$\Real^3$ to be a uniqueness set.
That it is sufficient follows from Lemma \ref{lem:Weak 2}.
\end{proof}
This result has a satisfying balance. The Core HK thm. \ref{thm:Core HK}
says that $\Real^3$ is a uniqueness set under the condition
that $\rho > 0$ everywhere,
but it does not say that the conclusion fails if the condition does.
It is not a clean dichotomy, but the Weak HK theorem is.

\subsubsection{consequences of (H1) and (H2)}
\label{sec:H1+H2}

We isolate the key argument in Lemma \ref{lem:DEU}, which parallels Lemma \ref{lem:H123}.
It uses the fact that intersection with $\Md(\Vint)$ preserves openness, connectedness and
denseness in $\Real^{3\Num}$, and it may be helpful to read the proof initially
assuming that $\vint\equiv 0$ so that $\Md(\Vint) = \Real^{3\Num}$.
\begin{lem}
\label{lem:DEU}
Assume (\ref{eq:H1}) and (\ref{eq:H2}), and suppose that
$W\in\comps \Md(\vext)$ intersects $\{\rho_1 > 0\}$.
Then, 
$\{\psi\neq 0\}$ is dense in $W\times W_2\times \cdots\times W_\Num$, for
some $W_2,\ldots,W_\Num \in \comps \Md(\vext)$.
\end{lem}
\begin{proof}
(H1) and (H2) imply that $\Md(\Vext)\cap\Md(\Vint)$ is dense in $\Real^{3\Num}$.
If $[W\times \Md(\vext)^{\Num-1}] \cap \Md(\Vint)$
were contained in the {\em closed} set $\{\psi=0\}$, so would its closure
$(\cl W)\times \Real^{3(\Num-1)}$, implying $W \subseteq \{\rho_1=0\}$,
contrary to hypothesis.
Hence, for some $W_2,\ldots,W_\Num \in \comps \Md(\vext)$, 
with $R \defeq W \times W_2\times\cdots\times W_\Num$,
$R \cap \Md(\Vint)$ is open, connected, and intersects $\{\psi \neq 0\}$.
Then, according to the Schechter-Simon UCP (\ref{thm:Schechter+Simon}),
$\{\psi\neq 0\}$ is dense in $R \cap \Md(\Vint)$, whence dense in $R$ itself.
\end{proof}

\begin{thm}[local Weak HK]
\label{thm:local Weak HK}
Assume (\ref{eq:H1}) and (\ref{eq:H2}).
Then, $W\in \comps \Md(\vext)$ is 
contained in either $\intr\{\rho=0\}$ or $\intr \cl\{\rho>0\}$.
In the second case, 
$\bigcup\setof{ U\in\comps\{\rho > 0\} }{ U\cap W \neq \varnothing }$
is a uniqueness set dense in $W$. 
\end{thm}
\begin{proof}
Suppose $W \not \subseteq \intr\{\rho=0\}$. 
Then, $W \cap \cl\{\rho_n > 0\}\neq \varnothing$ 
for some $n$; without loss, take $n=1$. 
Lemma \ref{lem:DEU} now gives a nonempty open set $R^\prime\subseteq \Real^{3(\Num-1)}$
such that 
$W\times R^\prime \subset \cl \{\psi\neq 0\}$,
and therefore $W \subseteq \intr \cl \{\rho> 0\}$.
That is the first conclusion. Continuing,
since $\{\psi \neq 0\} \subseteq \{\rho > 0\}^\Num$,
there must be $U_2,\ldots,U_\Num \in \comps \{\rho>0\}$ such that
$\varnothing \neq (W\times R^\prime) \cap (\Real^3 \times U_2\times \cdots \times U_\Num)
=W\times[R^\prime \cap (U_2\times \cdots \times U_\Num)]
\subset \cl \{\psi \neq 0\}$. 
Hence, if $U \in \comps \{\rho>0\}$ intersects $W$, then
$U\times U_2\times\cdots\times U_\Num$ intersects $\{\psi\neq 0\}$.
But, (i) such $U$ cover a dense subset of $W$ (by first conclusion) and (ii)
their union is a uniqueness set by Lemma \ref{lem:single particle transfer}. 
\end{proof}

\section{Strong}
\label{sec:Strong}

In Section \ref{sec:Weak}, we found conditions under which $c$-equivalence
could be carried through connected components of $\Md(\vext)$ as well as
of $\{\rho > 0\}$, and that if $\Md(\vext)$ is connected and dense,
$\{\rho > 0\}$ itself is a uniqueness set. 
Under the hypotheses (\ref{eq:H1}), (\ref{eq:H2}), and (\ref{eq:H3}) of the previous 
Section, we show that a point is a uniqueness point, regardless of whether 
$\rho$ vanishes there or not, if $\vint$ is locally $3\Num/2$ integrable on
a dense set and $\vext$ is locally $3\Num/2$ integrable near the point in 
question.
The key to this development is the strong UCP of Jerison \& Kenig discussed in 
Section \ref{sec:strong UCP}. 
The local Strong HK and Strong HK theorems are proven in Section \ref{sec:Strong thms}.
In view of the increasingly constraining nature of $3\Num/2$-integrability 
with $\Num$, Section \ref{sec:prospects} briefly considers prospects for
a better strong UCP.

\subsection{Jerison-Kenig strong UCP}
\label{sec:strong UCP}

A weak UCP concludes vanishing on the entire domain from vanishing
on an open set. 
One might expect that we now appeal to a result which reaches the
same conclusion from an assumption of vanishing on a nonzero measure set.
{\em Strong\/} UCPs, however, actually use the following notion.
A locally square integrable function $\psi$ is said to 
{\it vanish to infinite order at the point} $y$ if for every 
$N \in \Nat$,
\begin{equation}
\int_{B_r(y)} |\psi(x)|^2\, d^nx = {\cal O}(r^N) \quad\text{as}\quad r\to 0.
\end{equation}
Fortunately, as we discuss momentarily, this is a weaker assumption in our context 
than vanishing on a set of nonzero measure. 

The best strong UCP for potentials in $L^p$ spaces is 
\begin{thm}[Jerison \& Kenig, p. 479 of Ref.~\onlinecite{Jerison+Kenig}]
\label{thm:Jerison+Kenig}
Suppose, with $n\ge 3$, $q=2n/(n+2)$, and $\Omega$ a connected open subset of $\Real^n$,
that  $u\in W^{2,q}_{\tloc}(\Omega)$ [derivatives up to order $2$ are in $L^q_\tloc(\Omega)$]
satisfies (\ref{eq:diff-inequality}) for 
$V\in L^{n/2}_{\tloc}(\Omega)$. Then, if $u$ vanishes to infinite order
at any point, it is identically zero.
\end{thm}
In the context of many-body quantum mechanics, $n=3\Num$ is the dimension of
the configuration space. This is very unfortunate because $\vext$ 
and $\Vext$ sit at exactly the same point on the $L^p$ scale,
so that the larger the particle number $\Num$, the stronger the
restriction placed on $\vext$.
For $\Num \ge 2$, application of this theorem requires at least that
$\vext$ and $\vint$ be in $L^{3}_\tloc$, so these potentials are in $\Univ$.
Assuming the potentials are so, we must check that a solution $\psi$
of the Schr\"odinger eigenvalue equation (\ref{eq:Schrodinger})
satisfies the hypothesis $\psi \in W^{2,q}_\tloc(\Omega)$.
Since $q$ increases toward 2 as $n\to\infty$ ($\Num \to\infty$ in our application), 
it suffices to check
that $\nabla^2 \psi$ (which is equal to $\Vtot \psi$) is locally 
square integrable. But, $\psi$ is continuous, hence locally bounded. 
With $\Vtot \in L^{3\Num/2}_\tloc(\Omega)$, the hypothesis will
therefore be satisfied for $\Num \ge 2$.

It has been shown by de Figueiredo \& Gossez\cite{de-Figueiredo+Gossez-92},
as well as Regbaoui\cite{Regbaoui-01},
that for the $n$-dimensional Schr\"odinger equation with a potential in 
$L^{n/2}_{\tloc}(\Real^n)$, vanishing on a set $Z$ of nonzero measure implies vanishing 
to infinite order at almost every point of $Z$. 
This makes Thm. \ref{thm:Jerison+Kenig} relevant to our needs.
Prop. \ref{prop:zero-infinite-order} in Appendix \ref{sec:strong-implies-measure} 
presents a self-contained proof of a strengthening to the particular case which concerns 
us, namely that of a total potential derived from one- and two-body potentials. 
$\vext$ and $\vint$ are required only to be in $L^{3/2}_{\tloc}(\Real^3)$, 
indepedently of $\Num$. 
Prop. \ref{prop:zero-infinite-order} and Thm. \ref{thm:Jerison+Kenig} will
be applied via the Corollary following the next Definition.

\subsection{Sets of local ${3\Num/2}$-integrability}
\label{sec:lb-points}

\begin{defn}
\label{def:3N/2}
For $\Num\in\Nat$, the potential $v:\Real^3 \rightarrow \Real$ is 
{\it locally-$L^{3\Num/2}$ at} $x$ if there is an open neighborhood 
$U$ of $x$ such that $v \in L^{3\Num/2}(U)$.
The set of all points at which $v$ is {\it locally-$L^{3\Num/2}$} 
is denoted $\Tm_\Num(v)$. This is the largest set on which $v$ is
locally-$L^{3\Num/2}$ in the usual usage. [Thus, the standard term
``locally $L^p$'' is ``locally $L^p$ {\em everywhere}'' in 
current parlance.]
\end{defn}
Recall our standing assumption that $\Num \ge 2$. In that case,
$\Tm_\Num(v) \subseteq \Md(v)$ for $v=\vext$ or $\vint$.

\begin{cor}
\label{cor:measure UCP}
Assume $U \subseteq \Tm_\Num(\Vtot)$, and $\psi$ is a solution of the Schr\"odinger
equation (\ref{eq:Schrodinger}) in $U$ which is not identically zero.
Then, $\{\psi = 0\}\cap U$ has zero measure.
\end{cor}
\begin{proof}
Prop. \ref{prop:zero-infinite-order} shows that $\psi$ vanishes to infinite
order at almost every point of $\{\psi = 0\}$. Hence, if $\{\psi = 0\}$ has
nonzero measure, $\psi$ certainly vanishes to infinite order somewhere, and
therefore vanishes throughout $U$, according to Thm. \ref{thm:Jerison+Kenig}.
\end{proof}

\subsection{Strong HK theorems}
\label{sec:Strong thms}

To the hypotheses (\ref{eq:H1}) -- (\ref{eq:H3}) from the previous section,
we now consider adding two more:
\begin{align}
& \Tm_\Num(\vint) \;\text{is dense in}\; \Real^3
\label{eq:H4}
\tag{H4}
\\
& \Tm_\Num(\vext) \;\text{is of full measure}
\label{eq:H5}
\tag{H5}
  \end{align}
Paralleling what was done in the previous section, we will first establish 
that under hypothesis (\ref{eq:H4}), $\Tm_\Num(\vext)$ is a uniqueness set.
A lemma prepares the ground.
\begin{lem}
\label{lem:W-dense}
If $W\subseteq\Real^3$ is open and dense, then for given $x_1$, 
the set of $(x_2,\ldots,x_\Num)$ such that $x_n-x_m \in W$ for 
all $1\le m < n \le \Num$ is open and dense in  $\Real^{3(\Num-1)}$.
\end{lem}
\begin{proof}
$W_{n} = \setof{(x_2,\ldots,x_\Num)\in \Real^{3(\Num-1)}}{x_n - x_1 \in W}$ 
for $2 \le n \le \Num$, and
$W_{n,m} = \{(x_2,\ldots,x_\Num)\in \Real^{3(\Num-1)}\,|\, x_n - x_m \in W\}$ 
for $2\le m < n \le \Num$ 
are all clearly open and dense in $\Real^{3(\Num-1)}$,
hence so is their intersection. But the latter is precisely the required set.
\end{proof}
In the following theorem, this Lemma will be used under hypothesis (H4), with
$\Tm_\Num(\vint)$ in the r\^ole of $W$. Then it says that for any $x_1\in\Real^3$,
$\Tm_\Num(\Vint)$ is dense in the $3\Num-1$ dimensional fiber $\pi_1^{-1}(x_1)$
above $x_1$.

\begin{thm}[local Strong HK]
\label{thm:local Strong HK}
Assume (\ref{eq:H1}) -- (\ref{eq:H4}).
Then, $\Tm_\Num(\vext)$ is a uniqueness set.
{\it A fortiori}, 
$\Tm_\Num(\vext)$ is almost everywhere contained in $\{\rho > 0\}$
(equivalently, almost everywhere disjoint from $\{\rho=0\}$).
\end{thm}
\begin{proof}
By Lemma \ref{lem:Weak 2}, it suffices to show that 
$\rho > 0$ a.e. on a neighborhood of each point in $\Tm_\Num(\vext)$.
(That Lemma already says that $\rho > 0$ on a {\it dense} subset of $\Tm_\Num(\vext)$, 
so the issue is, loosely, one of closing the gap between `dense' and `almost everywhere'.)

Take \hbox{$x_1\in U \subseteq \Tm_\Num(\vext)$}, and consider the open set 
$U^\Num \defeq U\times\cdots \times U 
\subseteq \Tm_\Num(\Vext) \subseteq \Real^{3\Num}$.
As mentioned just before the statement of the Theomrem, 
Lemma \ref{lem:W-dense}, in conjunction with (H4) implies existence of 
some open connected set 
$W \subseteq U^\Num \cap \Tm_\Num(\Vint) \subseteq \Tm_\Num(\Vtot)$ satisfying
$x_1 \in \pi_1(W)$.

Now, by Lemma \ref{lem:H123}, $\{\psi\neq 0\}$ is dense in $\Real^{3\Num}$,
{\it a fortiori} in $W$. 
So, Cor. \ref{cor:measure UCP} implies that 
$\psi\neq 0$ almost everywhere on $W$. 
Therefore, by integration (over $x_2,\ldots,x_\Num$), 
$\rho_1 > 0$ a.e. on $\pi_1(W) \ni x_1$. 
\end{proof}
Under the hypotheses of the theorem, vanishing of $\rho$ on a set of
nonzero measure can be {\em locally} attributable to $\vext$ failing to be 
$3\Num/2$-integrable, since $\rho > 0$ almost everywhere on
$\Tm_\Num(\vext)$. This is similar in flavor to Thm. \ref{thm:Weak HK}
where it was shown that no point of $\bndy\intr\{\rho=0\}$ is in $\Md(\vext)$.

Now, to get full Hohenberg-Kohn style uniqueness, we need only assume
that $\Tm_\Num(\vext)$ is a set of full measure.
\begin{cor}[Strong HK]
\label{cor:Strong HK}
Assume (\ref{eq:H1}), (\ref{eq:H3}), (\ref{eq:H4}), and (\ref{eq:H5}).
Then, $\Real^3$ is a uniqueness set.
\end{cor}
\begin{proof}
$\Tm_\Num(\vext) \subseteq \Md(\vext)$, so hypothesis (H5) implies 
(H2). Thm. \ref{thm:local Strong HK} then applies.
\end{proof}

\subsection{Prospects for improvement}
\label{sec:prospects}

Insofar as there are no conditions on the density, the Strong 
HK Thm. \ref{cor:Strong HK} is the 
result closest to the traditional Hohenberg-Kohn theorem.
Since the total number of particles can certainly be determined from the 
density, a $\Num$-dependent condition does not violate the spirit of DFT.
However, the condition of local $3\Num/2$-integrability grows increasingly more 
stringent with $\Num$.
One might feel that the emphasis on $\Md(\vext)$ being merely dense, or
$\Tm_\Num(\vext)$ only of full measure are merely a pathetic attempt to squeeze
out a bit more generality. However, for the Strong HK theorem with large $\Num$,
it really matters. The measure zero exceptional set is very important to be
able to accomodate even isolated Coulomb singularities.
The possibility for improvement is unclear, since the $n/2$ exponent in the
Jerison-Kenig strong UCP is best possible. One would have to exploit the
gap between the differential inequality and differential equation, 
between vanishing on a set of nonzero measure and vanishing to infinite order,
or the special form (in configuration space) of the potential.

\section{Summary}
\label{sec:summary}

The Introduction bundled the three HK theorems of this paper into 
the somewhat vague form of the Omnibus HK theorem.
At this point, with the notions of $c$-equivalence (Def. \ref{def:uniqueness}),
Sobolev multiplier (Def. \ref{def:multiplier}),
and local $3\Num/2$-integrability (Def. \ref{def:3N/2}),
a more precise synthesis may be given by way of summary. 
In order to focus on the external potential $\vext$, 
which is anyway the more interesting, 
the following summary assumes that $\Tm_\Num(\vint)$
is dense and $\Md(\vint)$ is super-connected. The usual Coulomb
interaction satisfies these requirements easily, of course.
The Local HK theorems \ref{thm:local Core HK}, \ref{thm:local Weak HK} and
\ref{thm:local Strong HK} may be combined to yield these three assertions:
\begin{enumerate}
\item[(a)]
Each connected component of $\{\rho > 0\}$ is a uniqueness set.
\item[(b)]
If $\Md(\vext)$ is dense in $\Real^3$, then
the intersection of 
a connected component of $\{\rho > 0\} \cup \Md(\vext)$ 
with $\{\rho > 0\}$ 
is a uniqueness set.
That is, two points in $\{\rho > 0\}$ are $c$-equivalent if
they can be connected by a continuous path within $\{\rho > 0\} \cup \Md(\vext)$.
\item[(c)]
If, in addition, $\Md(\vext)$ is connected, then 
$\{\rho > 0\}$ is dense and 
$\{\rho > 0\} \cup \Tm_\Num(\vext)$ is a uniqueness set.
\end{enumerate}
From this, $\Real^3$ is a uniqueness set, i.e., full uniqueness in the 
sense of Hohenberg and Kohn holds relative to the universe $\Univ$, 
if any of the following holds:
\begin{enumerate}
\item[(a$^\prime$)]
$\rho > 0$ everywhere.
\item[(b$^\prime$)]
$\rho > 0$ almost everywhere,
$\Md(\vext)$ is dense,
and
$\{\rho > 0\} \cup \Md(\vext)$ is connected.
\item[(c$^\prime$)]
$\Md(\vext)$ is dense and connected,
and $\{\rho > 0\} \cup \Tm_\Num(\vext)$ has full measure.
\end{enumerate}
b$^\prime$ and c$^\prime$ are strengthenings of the Weak and Strong 
cases of the Omnibus HK theorem from the Introduction.

\acknowledgments{
I thank Jorge Sofo,
Vin Crespi, Thomas Hoffmann-Ostenhof and Markus Penz for 
comments and suggestions, and an anonymous referee whose
advice and questioning have led to vast improvements
in the paper.
}

\appendix


\section{vanishing on set of nonzero measure implies
a zero of infinite order}
\label{sec:strong-implies-measure}

In this appendix, we work in a general dimension $d$, rather than 3.
\begin{prop}
\label{prop:zero-infinite-order}
Suppose $\psi \in H_{\tloc}^1(\Real^{D})$ ($D \equiv \Num d$) 
is a solution of the Schr\"odinger equation (\ref{eq:Schrodinger}),
with $\vext,\vint \in L_{\tloc}^{d/2}(\Real^d)$.
Then, if $\psi$ vanishes on a set $Z$ of nonzero measure, it vanishes to
infinite order at almost every point of $Z$.
\end{prop}
The weak formulation is 
\begin{equation}
  \label{eq:weak-form}
\int \nabla\psi\cdot\nabla{\overline{\eta}} \, d^{D}x 
+ \int \Vtot \psi {\overline{\eta}} \, d^{D}x = 0, \quad \forall \eta\in C_c^\infty(\Real^{D}).
\end{equation}
The proof of the Proposition will use
\begin{lem}
\label{lem:doubling}
With everything as in the statement of Prop. \ref{prop:zero-infinite-order} and
$B_r(x)$ denoting the open ball of
radius $r$ about an arbitrary point $x\in \Real^{D}$ (later we drop the `$x$'),
for $r$ small enough,
  \begin{equation}
\int_{B_r(x)} |\nabla\psi|^2\, d^{D}x \le \frac{c(x)}{r^2} \int_{B_{2r}(x)} |\psi|^2\, d^{D}x.
  \end{equation}
\end{lem}
\begin{proof}
Take $h:\Real^{D}\rightarrow [0,1]$ a smooth bump function equal to 1 on $B_r(x)$ 
and supported on $B_{2r}(x)$, with $|\nabla h| \le 2/r$. 
In (\ref{eq:weak-form}), substitute $h^2\psi$ for $\eta$ (each integral 
is continuous in $\eta$ with respect to $H^{1}$ norm).
Then,
\begin{equation}
\label{eq:to-bnd-rhs}
\int |h\nabla\psi|^2 \, d^Dx= 
-2 \int h\nabla{\psi}\cdot\overline{\psi} \nabla h  \, d^Dx
- \int \Vtot |h\psi|^2\, d^Dx
\end{equation}
It is the left-hand side here that needs to be bounded.
Proceed by bounding each term on the right-hand side separately. The first is quickly
dispatched:
\begin{equation}
\label{eq:bound-1st-term}
\left| 2 \int h\nabla{\psi}\cdot\overline{\psi} \nabla h  \, d^Dx \right|
\le 
2\|h\nabla\psi\|_2 \|\psi\nabla h\|_2 
\le
\frac{1}{3}\|h\nabla\psi\|_2^2 + {3}\|\psi\nabla h\|_2^2.
\end{equation}
For the second integral in (\ref{eq:to-bnd-rhs}), recall that 
$\Vtot$ is a sum of $n = \Num(\Num+1)/2$ terms,
$\Num$ of which are of the form $\vext(x_i)$ and $\Num(\Num-1)/2$ of the form $\vint(x_j-x_i)$.
We show how to handle $\vext(x_1)$.
Define
\begin{equation}
\tilde{\rho}(x_1) = \int_{\Real^{(\Num-1)d}} |h\psi|^2 dx_2\cdots dx_\Num.
\label{eq:tilde rho}
\end{equation}
Putting this into the integral in (\ref{eq:to-bnd-rhs}),
split the integration domain according to whether $|\vext(x_1)|$ is smaller or
larger than some constant $M$, to be chosen later.
Then, bound the contribution of $\vext(x_1)$ to the integral by
\begin{equation}
M\int_{|\vext(x)|<M}\tilde{\rho} \; d^dx 
 + \int_{|\vext(x)|\ge M} |v|\tilde{\rho}\; d^dx
\le 
M\|h\psi\|_2^2 
+ \| \vext(x_1)\cdot \chi(B_{2r}(x_1))\, \chi(|\vext(x_1)|\ge M)\|_{\frac{d}{2}} \|\tilde{\rho}\|_{\frac{d}{d-2}},
\label{eq:B-bound}
\end{equation}
where the last term was obtained by use of the H\"older inequality,
and $M$ is yet to be chosen.
The other $n-1$ terms are handled in nearly the same way. For instance,
the analog of (\ref{eq:tilde rho}) for $\vint(x_2-x_1)$ holds $x_1-x_2$ fixed.

Now, for the second factor in the last term of (\ref{eq:B-bound}), use
\begin{equation}
\|\tilde{\rho}\|_{\frac{d}{d-2}} = 
\|\tilde{\rho}^{1/2} \|_{\frac{2d}{d-2}}^2 \le C \|\nabla\tilde{\rho}^{1/2}\|_2^2 
\le C \|\nabla(h\psi)\|_2^2.
\end{equation}
The first inequality is a Sobolev inequality [e.g., Thm. V.5.4 of Ref. \onlinecite{Adams}, or
Thm. 5.26 of Ref. \onlinecite{Robinson-Infinite-Dim-Dyn-Sys}], and the second follows, 
as in Thm. 1.1 of Ref.~\onlinecite{Lieb83}, from 
\begin{equation}
|\nabla\tilde{\rho}| =
\left| 2\re\, \int \nabla(h\psi) (h\overline{\psi}) dx_2\cdots dx_\Num \right|
\le 2|\tilde{\rho}^{1/2}| \left( \int |\nabla(h\psi)|^2 dx_2\cdots dx_\Num \right)^{1/2}
\end{equation}
upon division of both sides by $|\tilde{\rho}^{1/2}|$, squaring and integrating.

Turning to the first factor in the last term of (\ref{eq:B-bound}), 
\hbox{$\| \vext(x_1) \cdot \chi(B_{2r}(x_1)) \chi(|\vext(x_1)|\ge M)\|_{{d}/{2}}$},
choose $M$ large enough (depending on $x$) that it, 
as well as the $n-1$ similar factors arising from the other terms of $\Vtot$
are all less than $1/(6n C)$. This is possible since $\vext, \vint \in L^{d/2}_\tloc$.
All together, then,
\begin{equation}
\label{eq:bound-2nd-term}
\left| \int \Vtot |h\psi|^2 d^Dx \right|
\le n M \|h \psi\|_2^2 
+ \frac{1}{6} \|\nabla(h \psi)\|_2^2
\le n M \|h \psi\|_2^2 
+ \frac{1}{3} \|h\nabla\psi\|_2^2 + \frac{1}{3} \|\psi\nabla h\|_2^2.
\end{equation}
Inserting the bounds (\ref{eq:bound-1st-term}) and (\ref{eq:bound-2nd-term}) 
into (\ref{eq:to-bnd-rhs}) yields
\begin{equation}
\|h\nabla\psi\|_2^2  
\le
{10} \|\psi\nabla h\|_2^2  + 
3 n M  \|h \psi\|_2^2.
\end{equation}
Finally, by use of $|\nabla h| \le 2/r$,
\begin{equation}
\int_{B_r(x)} |\nabla\psi|^2 \, d^{D}x 
\le
\int h^2 |\nabla\psi|^2 \, d^{D}x 
\le
\left(3n M + \frac{40}{r^2}\right)
\int_{B_{2r}(x)} |\psi|^2 \, d^{D}x.
\end{equation}
\end{proof}

\begin{proof}[Proof of Prop. \ref{prop:zero-infinite-order}]
Recall that $Z$ denotes the set where $\psi=0$, assumed of nonzero
Lebesgue measure. Almost every point $x\in Z$ is a point of Lebesgue density, 
which means that
\begin{equation}
\lim_{r\to 0} \frac{|Z \cap B_r(x)|}{|B_r(x)|} = 1.
\end{equation}
It is convenient to express this in the form
\begin{equation}
{|Z^c \cap B_r(x)|} \le [r K(r)]^{D}
\end{equation}
for some monotonic function $K:\Real_+ \rightarrow \Real_+$ with $K(0)=0$.
Then, by H\"older's inequality,
\begin{equation}
\label{eq:estimate-2}
  \int_{B_r(x)} |\psi|^2 \, d^Dx= 
  \int_{B_r(x)\cap Z^c} |\psi|^2 \, d^Dx
\le
  \left( \int_{B_r(x)\cap Z^c} |\psi|^{\frac{2D}{D-2}} \, d^Dx\right)^{\frac{D-2}{D}} 
|B_r(x)\cap Z^c|^{\frac{2}{D}}
\end{equation}
To estimate the integral on the right-hand side, we apply
an extension theorem
[see, for example, Thm. IV.4.26 of Ref.~\onlinecite{Adams},
or Thm. 5.20 of Ref. \onlinecite{Robinson-Infinite-Dim-Dyn-Sys}]
which says that there is a constant $C$ 
such that every $f\in H^1(B_1(0))$ 
has an extension to $\tilde{f} \in H^1(\Real^D)$
satisfying $\|\tilde{f}\|_{H^1(\Real^D)} \le C \|f\|_{H^1(\Real^D)}$.
Combined with a Sobolev inequality, this yields
\begin{equation}
\left( \int_{B_1(0)} |f|^{\frac{2D}{D-2}} \,d^Dx \right)^{\frac{D-2}{D}} 
\le C \|f\|_{H^1(B_1(0))}.
\end{equation}
Now apply this to $f(y) = \psi(r(x-y))$, using $d^Dy = r^{-D} d^D(r(x-y))$
and $\nabla f(y) = r (-\nabla\psi)(r(x-y))$, to find
\begin{equation}
\left( \int_{B_r(x)\cap Z^c} |\psi|^{\frac{2D}{D-2}} \, d^Dx \right)^{\frac{D-2}{D}} 
\le C \left(
\int_{B_r(x)}|\nabla\psi|^2\, d^Dx + \frac{1}{r^2}\int_{B_r(x)}|\psi|^2 \, d^Dx \right).
\end{equation}
Substituting this back into (\ref{eq:estimate-2}) and then applying Lemma \ref{lem:doubling}
results in
\begin{align}
  \int_{B_r(x)} |\psi|^2 \, d^Dx
& \le C K(r)^2 \left(r^2\int_{B_r(x)}|\nabla\psi|^2 \, d^Dx + \int_{B_r(x)}|\psi|^2 \, d^Dx \right)
\nonumber \\
& \le C(1+c(x)) K(r)^2 \int_{B_{2r}(x)}|\psi|^2 \, d^Dx.
\label{eq:doubling}
\end{align}
It is now a straight line from this inequality to the desired conclusion,
$\int_{B_r(x)}|\psi|^2 = {\cal O}(r^N)$ as $r \to 0$.
With the definitions
\begin{equation}
F(r) \defeq \int_{B_r(x)}|\psi|^2\, d^Dx,\quad g(2r) \defeq C(1+c(x)) K(r)^2,
\end{equation}
an induction starting from (\ref{eq:doubling}) yields
$F(2^{-n}r) \le g(r)^{n}{F(r)}$ for $n\in\Nat$. Thus,
\begin{equation}
\frac{F(2^{-n}r)}{(2^{-n}r)^N}
\le
\left({2^{N}}{g(r)}\right)^n \frac{F(r)}{r^N}.
\end{equation}
Choosing $r$ small enough that $2^N g(r) < 1$, the right-hand
side tends to zero as $n\to\infty$. Appeal to monotonicity
[to cover the intervals $(2^{-(n+1)}r,2^{-n}r)$] now gives the required conclusion:
$F(x) = {\cal O}(x^N)$ as $x\to 0$ for all $N\in\Nat$.
\end{proof}


%

\end{document}